\documentclass[aip,cha,amsmath,amssymb,preprint]{revtex4-1}

\usepackage{xcolor}
\usepackage{bm}
\usepackage{amsthm}

\usepackage{physics}
\usepackage{mathtools}
\mathtoolsset{showonlyrefs}

\theoremstyle{plain}
\newtheorem{theorem}{Theorem}[section]
\newtheorem{proposition}[theorem]{Proposition}

\newtheorem{corollary}[theorem]{Corollary}
\theoremstyle{definition}
\newtheorem{definition}[theorem]{Definition}

\theoremstyle{remark}
\newtheorem{remark}[theorem]{Remark}

\newcommand{\cH}{\mathcal H}
\newcommand{\cL}{\mathcal L}
\newcommand{\cE}{\mathcal E}
\newcommand{\cD}{\mathcal D}
\newcommand{\cR}{\mathcal R}
\newcommand{\cK}{\mathcal K}
\newcommand{\cP}{\mathcal P}
\newcommand{\cS}{\mathcal S}
\newcommand{\cN}{\mathcal N}
\newcommand{\cB}{\mathcal B}
\newcommand{\cM}{\mathcal M}
\newcommand{\id}{\mathrm{id}}

\begin{document}

\title{Code-Visible Liouvillian Modes and Coherent Logical Errors in Quantum Codes}
\author{Marina Gordeychuk}
\email{m.lisnichenko@innopolis.university}
\author{Oleg Kiselev}
\email{o.kiselev@innopolis.ru}
\affiliation{Innopolis University, Russia}

\date{10 July 2026}

\begin{abstract}
Coherent Hamiltonian drifts and systematic errors can pass through quantum
error correction differently than stochastic Pauli noise: they generate
phase components of the logical channel that are lost under full Pauli
twirling. We consider the Lindbladian dynamics of \(n\) physical qubits
propagated through an effective logical channel of a fixed encoding,
recovery, and decoding scheme. The work is methodological in character:
standard ideas of observability and minimal realization from linear systems
theory are applied to the channel
\[
    \omega
    \mapsto
    \Phi_t(\omega)
    =
    \cD e^{t\cL}\cE(\omega),
\]
where \(\cE\) encodes the logical state and \(\cD\) is the chosen recovery
map with final decoding. The minimal code-visible space \(\cK_{\rm code}\)
is defined as the Krylov closure of the decodable logical observables with
respect to \(\cL^\dagger\); its observable factorization yields a minimal
realization of the logical channel. In this form, the Knill--Laflamme
conditions serve as a short-time consistency check, Pauli noise on
stabilizer codes reduces to a chain on syndrome-logical classes, and the
choice of recovery map acts as a spectral filter on the visible modes. For
quantum memory, the same language yields an estimate of the storage time in
terms of the dominant visible eigenvalues of the logical channel of the
correction cycle. Particular attention is paid to the spectral signatures
of coherent noise and to how these signatures depend on the chosen decoder.
\end{abstract}


\maketitle 

\section{Introduction}

A quantum error-correcting code preserves logical information embedded in a
designated code subspace or subsystem, up to the chosen recovery procedure.
To analyze the decoherence of codes, we distinguish the full physical
trajectory of the state in the \(2^n\)-dimensional space, the syndrome
information that a fixed code can extract, and the effective logical
channel obtained after recovery and decoding. It is precisely the latter
object that is operationally observable for a memory or a logical block of
computation.

For quantum memory this reduction is especially natural. The quality of the
memory is measured by the closeness of the storage logical channel to the
identity channel. Under periodic recovery, the long-time behavior is set by
the spectrum of the single-cycle logical channel. The actual storage time
of logical information is determined by those physical modes that survive
encoding, recovery, and decoding.

A separate motivation comes from coherent errors. In threshold calculations
and in fast simulations, noise is often replaced by a Pauli channel, e.g.,
via Pauli twirling or randomized compiling
\cite{WallmanEmerson2016,KatabarwaGeller2015,CaiXuBenjamin2020}. Such a
replacement is useful, but it erases the phase relations between the Pauli
components of the noise. For a systematic Hamiltonian drift, these phases
can turn into logical rotations and non-monotonic logical fidelity; on
surface-code models this manifests as a difference between coherent
physical noise and its Pauli-twirled approximation
\cite{BravyiEnglbrechtKoenigPeard2018}. This calls for a tool that starts
from the continuous Lindblad/GKSL generator and shows which of its spectral
components actually reach the logical channel after a chosen recovery map.

In the classic works of Shor, Steane, Calderbank--Shor, and
Knill--Laflamme \cite{Shor1995,Steane1996,CalderbankShor1996,KnillLaflamme1997}
the idea of encoding quantum information in a larger Hilbert space was
formulated, and conditions for exact recovery were obtained. In
Gottesman's stabilizer language \cite{Gottesman1997} these conditions take
an especially transparent form: an error is identified by its syndrome, and
the residual after correction is either a stabilizer or a nontrivial
logical operator. The relation of quantum codes to channels, entanglement,
and distillation was developed in \cite{Bennett1996}. General Markovian
decoherence is described by Lindblad/GKSL generators
\cite{GKS1976,Lindblad1976}, and decoherence-free subspaces give a limiting
case of a code for which the corresponding class of errors has a trivial
logical image \cite{ZanardiRasetti1997,LidarChuangWhaley1998}. The
operator and algebraic formulations of quantum error correction in the
Heisenberg picture were developed in
\cite{BenyKempfKribs2007,BenyKempfKribsPRA2007}, and conditions for
correctability under continuous Lindblad/Hamiltonian dynamics were studied
in \cite{OreshkovLidarBrun2008}. On the other hand, the linear-algebraic
idea of an observable subspace and minimal realization goes back to
classical control theory \cite{Kalman1960,HoKalman1966}.

The subject of this work is the spectral reduction of a given code,
recovery map, and physical Lindbladian generator to an effective logical
channel. The full Liouvillian acts on the \(4^n\)-dimensional space of
operators. The logical channel for \(k\) qubits is described by at most
\(16^k\) real parameters. This dimensional gap gives rise to a natural
projection of physical modes onto logically distinguishable modes.

We call such a mode a \emph{code-visible Liouvillian mode}: it is a
spectral component of the physical dynamics that gives a nonzero
contribution to at least one matrix element of the effective logical
channel.

Components that change only unused syndrome coherences, states outside the
recovery domain, or physical degrees of freedom discarded by the decoder
belong to the kernel of the logical projection. Logical decoherence for a
given code is set precisely by the components that change the channel
\(\Phi_t\).

The contribution of this work is a QEC specialization of this standard
linear-systems apparatus. First, a finite algorithm is given for
constructing the minimal observable realization \(\widehat{\cK}_{\rm
code}\) from the matrices \(\cL^\dagger,V,\cD\); this replaces the full
evolution on the \(4^n\)-dimensional operator space with a smaller matrix
that contains exactly the logically visible dynamics. Second, the role of
the recovery map is formulated as that of a spectral filter, and the
mismatch between the decoder and coherent noise becomes visible at the
level of complex logical modes. Third, the same reduction is applied to
quantum memory, where the storage time is expressed through the dominant
eigenvalues of the single-cycle logical channel. The examples below are
chosen to be minimal: they serve as a check of the formalism and separate
the contribution of the recovery map from already known correctability
conditions.

We collect the main elements of the work below.

Let \(\cL\) be a finite-dimensional Lindblad/GKSL generator on the physical
qubits, \(\cE\) a fixed encoding of a logical system \(\cH_L\), and \(\cD\)
a fixed recovery map with decoding. Set
\[
    \Phi_t=\cD e^{t\cL}\cE .
\]
Let \(\{F_j\}_{j=1}^{d_L^2}\) be a basis of logical observables on
\(\cH_L\), \(d_L=\dim\cH_L\), and
\[
    \cK_{\rm code}
    =
    \operatorname{span}_{\mathbb C}
    \{(\cL^\dagger)^m \cD^\dagger(F_j)\colon
    j=1,\ldots,d_L^2,\ m\ge0\}.
\]

Then the effective logical dynamics \(\Phi_t\) is completely determined by
the restriction \(A=\cL^\dagger|_{\cK_{\rm code}}\), and the space
\(\cK_{\rm code}\) together with its truly observable factorization is
constructed by a finite Krylov-closure algorithm, yielding an effective
model of logical decoherence without integrating the full physical
dynamics (Proposition \ref{prop:visible-subspace} and Theorem
\ref{thm:krylov-visible-algorithm}). The standard Jordan decomposition of
this minimal realization gives three elementary types of visible spectral
contributions to logical decoherence --- aperiodic, damped oscillatory,
and undamped oscillatory --- where sinusoidal components of the logical
channel arise if and only if \(A\) has a code-visible complex conjugate
pair of eigenvalues, and for a Hamiltonian drift such pairs are a
spectral signature of a coherent logical error (Theorem
\ref{thm:strict-visible-classification}). The Knill--Laflamme conditions
for Lindblad errors give a short-time check of the reduction: the
recovered logical channel has zero dissipative first order in time, and
the remaining first order is a logical Hamiltonian drift in the absence
of compensation (Proposition \ref{prop:first-order-correction}). For a
stabilizer code and Pauli-Lindblad noise, the dynamics reduces to a
finite chain on Pauli classes, and after recovery, to a logical Pauli
channel, where under detailed balance damped oscillations are absent in
this classical reduction (Theorem \ref{thm:pauli-chain-reduction} and
Corollary \ref{cor:detailed-balance-no-oscillations}); and for a recovery
map obtained by logical post-processing, the visible space is contained in
the original one: incomparable recovery maps can see different parts of
the same Liouvillian, which gives a formal model of decoder mismatch for
coherent noise (Theorem \ref{thm:recovery-comparison} and Example
\ref{prop:same-drift-two-recoveries}). A quantum memory with periodic
recovery is described by powers of the single-cycle logical channel,
where the dominant visible eigenvalue modulus of this channel sets the
storage time, and complex modes give oscillations and revivals of
fidelity (Theorem \ref{thm:memory-visible-modes} and Proposition
\ref{prop:error-correction-memory-scaling}); finally, the dimension
of the code-distinguishable asymptotics is bounded above by
\(d_L^4-d_L^2\), i.e., by the dimension of the affine space of CPTP
channels on the logical system (Theorem
\ref{thm:asymptotic-code-channels}).

\begin{remark}
Below, the word ``classification'' refers to the fixed code projection
\(\Phi_t=\cD e^{t\cL}\cE\). The object of classification is those parts of
the spectrum of the physical Lindblad/GKSL generator that the chosen code
and the chosen recovery map turn into logical decoherence.
\end{remark}

\section{Geometry of a quantum error-correcting code}
\label{sec:code-geometry}

Let \(\cH_L\) be a logical Hilbert space of dimension
\(d_L=2^k\), and let
\[
    \cH_P=(\mathbb C^2)^{\otimes n}
\]
be the space of \(n\) physical qubits. A subspace code is given by an
isometry
\[
    V:\cH_L\to\cH_P,
    \qquad
    V^\dagger V=I_L.
\]
The code projector is
\[
    P=VV^\dagger.
\]

Encoding as a quantum channel has the form
\begin{equation}
    \cE(\omega)=V\omega V^\dagger,
    \qquad
    \omega\in\cB(\cH_L).
\end{equation}
Recovery together with decoding will be described by a single CPTP channel
\[
    \cD:\cB(\cH_P)\to\cB(\cH_L).
\]

One may think of \(\cD\) as first measuring the syndrome and applying a
correction, and then identifying the code subspace with \(\cH_L\). This
combination is convenient because all operationally significant
quantities depend only on the composition
\[
    \Phi=\cD\cN\cE,
\]
where \(\cN\) is the physical noise channel.

\begin{definition}
For a physical channel \(\cN\), the effective logical channel of the code
\((\cE,\cD)\) is defined as
\begin{equation}
    \Phi_{\cN}^{\rm code}
    :=
    \cD\cN\cE .
\end{equation}
If the physical noise is given by a semigroup
\(\cN_t=e^{t\cL}\), we write
\begin{equation}
    \Phi_t
    =
    \cD e^{t\cL}\cE .
\end{equation}
\end{definition}

Decoherence of the code is described by the change of the logical channel
\(\Phi_t\). Two physical processes are considered logically
indistinguishable if, after recovery and decoding, they give the same
channel on \(\cH_L\).

\subsection{Knill--Laflamme conditions}

Suppose the physical noise has a set of errors \(E_a\). A subspace code
with projector \(P\) corrects this set if and only if there exist numbers
\(c_{ab}\) such that
\begin{equation}
\label{eq:knill-laflamme}
    P E_a^\dagger E_b P=c_{ab}P .
\end{equation}

These are the Knill--Laflamme conditions \cite{KnillLaflamme1997}. Their
geometric meaning is that the syndrome information is determined by the
error, and the overlap matrix of the error subspaces is the same for all
logical states within the code.

\begin{proposition}[Logical invisibility of correctable errors]
\label{prop:kl-invisibility}
Let \(\{E_a\}\) satisfy the Knill--Laflamme conditions
\eqref{eq:knill-laflamme}, and let \(\cN\) be any CPTP channel on
\(\cB(\cH_P)\) whose Kraus operators lie in the linear span of
\(\{E_a\}\). Then there exists a recovery map \(\cD\) such that
\begin{equation}
    \cD\cN\cE=\id_L .
\end{equation}
\end{proposition}

\begin{proof}
This is the standard form of the correctability criterion. Condition
\eqref{eq:knill-laflamme} means that the subspaces \(E_aP\cH_P\) have the
same overlap matrix for all logical states. After diagonalizing the matrix
\(c_{ab}\), one can choose an equivalent set of errors that map the code
subspace into mutually orthogonal syndrome subspaces. Measuring the
syndrome and applying the inverse isometry on each such subspace recovers
the original \(\omega\).
\end{proof}

\begin{remark}
The Knill--Laflamme conditions are static: they specify a correctable
discrete set of operators. Under continuous Lindblad dynamics, this is
supplemented by the rates at which errors occur and by the spectral modes
of the physical generator that remain visible after \(\cD\).
\end{remark}

\subsection{Stabilizer geometry of syndromes}

For a stabilizer \([[n,k,d]]\) code, one chooses an abelian subgroup
\(\cS\) of the Pauli group \(\cP_n\) not containing \(-I\). The code space
is the common \(+1\)-eigenspace of the stabilizers:
\[
    \cH_C
    =
    \{\,\psi\in\cH_P\colon S\psi=\psi,\ S\in\cS\,\}.
\]
Stabilizers are Pauli operators that leave every valid code state
unchanged. Hence an error-free state yields the outcome \(+1\) when
measuring each stabilizer. If an error anticommutes with some stabilizer,
the sign flips; the collection of these signs is the syndrome. The
syndrome of a Pauli operator \(E\in\cP_n\) is defined by the commutation
signs with the independent stabilizer generators
\[
    S_1,\ldots,S_{n-k}.
\]
If \(ES_j=(-1)^{s_j}S_jE\), then
\[
    s(E)=(s_1,\ldots,s_{n-k})\in\mathbb F_2^{n-k}.
\]
Choosing a recovery map means choosing a Pauli operator \(R_s\) for each
syndrome \(s\). After an error \(E\) and correction \(R_{s(E)}\), the
residual operator
\[
    R_{s(E)}E
\]
commutes with all stabilizers, i.e., it belongs to the normalizer
\(\cN(\cS)\). Its class in the quotient
\[
    \cN(\cS)/\cS
\]
is a logical Pauli operator. Stabilizer recovery uses the syndrome of the
error together with its logical residue.

\begin{definition}
For a fixed choice of recovery \(R_s\), define the logical residue of a
Pauli error \(E\) as
\begin{equation}
\label{eq:logical-residue}
    \ell(E)
    =
    [R_{s(E)}E]\in \cN(\cS)/\cS .
\end{equation}
The error is corrected successfully if and only if \(\ell(E)\) is the
trivial class.
\end{definition}

This formula is the stabilizer analogue of Bell-population statistics in
teleportation: the full physical error \(E\) is projected onto a finite
logical class \(\ell(E)\), and it is precisely this class that determines
the logical channel after recovery. A physical error acts on the actual
qubits. A logical error is the residual action on the encoded information
after syndrome measurement and correction. If the residual class is
trivial, the logical state is preserved; if the class is nontrivial, the
error becomes a nontrivial logical Pauli operator (for
\(k=1\) logical qubit, a logical \(X\), \(Y\), or \(Z\); for \(k>1\), a
nontrivial logical Pauli string).

\section{Lindbladian formulation for a code}
\label{sec:lindblad-code}

Let the physical evolution be given by
\begin{equation}
    \dot\rho(t)=\cL(\rho(t)),
\end{equation}
where
\begin{equation}
\label{eq:lindblad}
    \cL(\rho)
    =
    -i[H,\rho]
    +
    \sum_\alpha
    \left(
        L_\alpha\rho L_\alpha^\dagger
        -
        \frac12\{L_\alpha^\dagger L_\alpha,\rho\}
    \right).
\end{equation}
For an encoded initial state \(\rho(0)=V\omega V^\dagger\), the
operationally significant trajectory is
\[
    \Phi_t(\omega)=\cD(e^{t\cL}(V\omega V^\dagger)).
\]
Here \(\rho(t)\) is the physical state of all \(n\) qubits, while
\(\Phi_t(\omega)\) is the logical state after noise, recovery, and
decoding. It is \(\Phi_t\) that measures the quality of memory or of a
logical computational block. This setup separates stability tasks at
different levels: preserving the physical code state \(V\omega
V^\dagger\), preserving the code subspace as a set, preserving the
syndrome-correctable region, and preserving the effective logical channel
\(\Phi_t\). These are different levels of description. For instance, a
correctable error can move the state out of the code subspace, while after
recovery the logical channel remains the identity.

\subsection{Relation to observability of linear systems}

The following construction adapts finite-dimensional observability
\cite{Kalman1960,HoKalman1966} to the channel \(\Phi_t\). If one
vectorizes the space of physical operators, the
adjoint Lindblad dynamics takes the form
\[
    \dot X=\cL^\dagger X .
\]
The output coordinates in this problem are fixed by the matrix elements of
the recovered logical channel:
\[
    y_{j,\omega}(t)
    =
    \operatorname{Tr}
    \left[
        e^{t\cL^\dagger}\cD^\dagger(F_j)\,
        V\omega V^\dagger
    \right].
\]

In the classical notation \(\dot x=Ax,\ y=Cx\), the observable part is
given by the linear span of \(C,CA,CA^2,\ldots\), or, in adjoint form, by
the Krylov closure of the initial output observables. In the present
problem the role of these initial observables is played by the operators
\(\cD^\dagger(F_j)\), while the compression onto the code is given by the
map
\[
    \Gamma(X)=V^\dagger X V .
\]

The map \(\Gamma\) takes a physical observable \(X\) and looks only at its
action inside the code subspace. Hence \(\Gamma(X)=0\) means: the
observable \(X\) may exist in the full physical system, but the encoded
logical block does not see it. This is why the terms ``code-visible mode''
and ``unobservable subspace'' should be understood as a QEC specialization
of the standard notions of observability. The physical content of this
specialization is that the output map is set by the chosen encoding,
recovery map, and decoding.

\subsection{Code-visible space of observables}

Passing to the adjoint generator lets us immediately isolate the visible
part of the dynamics. Let \(\{F_j\}_{j=1}^{d_L^2}\) be a basis of
\(\cB(\cH_L)\). For any logical state \(\omega\),
\begin{equation}
\label{eq:heisenberg-code-pairing}
    \operatorname{Tr}
    \left[
        F_j \Phi_t(\omega)
    \right]
    =
    \operatorname{Tr}
    \left[
        e^{t\cL^\dagger}\cD^\dagger(F_j)\,
        V\omega V^\dagger
    \right].
\end{equation}
Hence all matrix elements of the logical channel are determined by the
evolution of the observables \(\cD^\dagger(F_j)\).

\begin{definition}
The minimal code-visible invariant space is defined as
\begin{equation}
\label{eq:k-code-definition}
    \cK_{\rm code}
    =
    \operatorname{span}_{\mathbb C}
    \{(\cL^\dagger)^m \cD^\dagger(F_j)\colon
    j=1,\ldots,d_L^2,\ m\ge0\}.
\end{equation}
\end{definition}

\begin{proposition}[Minimal visible reduction]
\label{prop:visible-subspace}
The space \(\cK_{\rm code}\) is finite-dimensional, invariant under
\(\cL^\dagger\), and the effective logical channel \(\Phi_t\) is
completely determined by the restriction
\[
    A=\cL^\dagger|_{\cK_{\rm code}}.
\]
More precisely, if for all \(j\)
\[
    e^{t\cL^\dagger}\cD^\dagger(F_j)
    =
    e^{tA}\cD^\dagger(F_j),
\]
then all functions
\[
    \operatorname{Tr}[F_j\Phi_t(\omega)]
\]
are obtained by pairing the right-hand side of
\eqref{eq:heisenberg-code-pairing}.
\end{proposition}

\begin{proof}
In finite dimension, the linear span of all operators
\[
    (\cL^\dagger)^m\cD^\dagger(F_j)
\]
is finite-dimensional. By construction it is invariant under
\(\cL^\dagger\). Therefore
\[
    e^{t\cL^\dagger}\cD^\dagger(F_j)
\]
can be computed inside \(\cK_{\rm code}\) using the matrix
\(A=\cL^\dagger|_{\cK_{\rm code}}\). Formula
\eqref{eq:heisenberg-code-pairing} shows that this suffices for all
matrix elements of the logical channel.
\end{proof}

\begin{theorem}[Finite algorithm for the visible reduction]
\label{thm:krylov-visible-algorithm}
Given a matrix representation of \(\cL^\dagger\), a code isometry \(V\),
and a recovery map \(\cD\), the code-visible reduction is constructed in a
finite number of linear-algebraic steps. First choose any basis of
logical observables \(F_1,\ldots,F_{d_L^2}\) and the initial space
\[
    W_0=\operatorname{span}\{\cD^\dagger(F_j)\colon
    j=1,\ldots,d_L^2\}.
\]
Then successively build the spaces
\[
    W_{r+1}=W_r+\cL^\dagger(W_r),
\]
stopping at the first \(r\) for which
\(W_{r+1}=W_r\), so that \(W_r=\cK_{\rm code}\), the stopping occurring
within at most \(d_P^2\) steps, where
\(d_P=\dim\cH_P=2^n\). Next, on \(\cK_{\rm code}\) take the matrix
\(A=\cL^\dagger|_{\cK_{\rm code}}\) and introduce the compression map onto
the code
\[
    \Gamma:\cK_{\rm code}\to\cB(\cH_L),
    \qquad
    \Gamma(X)=V^\dagger X V ,
\]
after which the maximal unobservable subspace equals
\begin{equation}
\label{eq:unobservable-subspace}
    \cN_{\rm obs}
    =
    \{X\in\cK_{\rm code}\colon
    \Gamma(A^mX)=0,\quad
    m=0,\ldots,\dim\cK_{\rm code}-1\}
\end{equation}
and is invariant under \(A\); finally, the minimal observable realization
of the logical channel is given by the quotient space
\[
    \widehat{\cK}_{\rm code}
    =
    \cK_{\rm code}/\cN_{\rm obs}
\]
together with the induced operator \(\widehat A\).

Two vectors in \(\cK_{\rm code}\) give the same contribution to all matrix
elements of \(\Phi_t\) for all \(t\ge0\) if and only if their difference
lies in \(\cN_{\rm obs}\).
\end{theorem}

\begin{proof}
The first two points are an ordinary Krylov closure of the initial
observables with respect to \(\cL^\dagger\). Since the space of operators
on \(\cH_P\) has dimension \(d_P^2\), the sequence \(W_r\) stabilizes
within at most \(d_P^2\) strict extensions. The condition
\(W_{r+1}=W_r\) means invariance under \(\cL^\dagger\), and minimality
follows from the fact that any invariant space containing all
\(\cD^\dagger(F_j)\) must contain all their images
\((\cL^\dagger)^m\cD^\dagger(F_j)\).

An operator \(X\in\cK_{\rm code}\) has zero logical image if
\[
    \Gamma(e^{tA}X)=0
\]
for all \(t\). In finite dimension this is equivalent to the vanishing of
all derivatives at zero up to order \(\dim\cK_{\rm code}-1\), i.e., to
conditions \eqref{eq:unobservable-subspace}; higher powers of \(A\) are
expressed through the previous ones by the Cayley--Hamilton theorem. The
same formula gives the invariance of \(\cN_{\rm obs}\): for
\(X\in\cN_{\rm obs}\), the quantities \(\Gamma(A^mAX)\) vanish for all the
needed \(m\), and the last case again reduces to Cayley--Hamilton.
Therefore the quotient space is well defined and identifies exactly those
variables whose contribution to \(\Phi_t\) is identically zero.
\end{proof}

\begin{corollary}[Criterion for a constant logical channel]
\label{cor:logical-channel-constant}
The logical channel \(\Phi_t\) remains constant for all logical initial
states if and only if for all \(j\) and all \(\omega\)
\[
    \operatorname{Tr}
    \left[
        \left(e^{t\cL^\dagger}\cD^\dagger(F_j)-\cD^\dagger(F_j)\right)
        V\omega V^\dagger
    \right]
    =
    0 .
\]
A sufficient strong condition is
\[
    \cL^\dagger\cD^\dagger(F_j)=0,
    \qquad
    j=1,\ldots,d_L^2.
\]
\end{corollary}

The latter condition is a strong sufficient criterion. A more general
mechanism for preserving the logical channel is the vanishing of the
change of the observables \(\cD^\dagger(F_j)\) on the encoded states after
the chosen recovery map.

\subsection{Decoder mismatch for coherent noise}

The code-visible space is determined by the code, the physical noise, and
the chosen recovery map. Different recovery maps can identify or
distinguish syndrome regions in which the physical dynamics carries phase
information. For stochastic Pauli noise this usually shows up as a
difference in the probabilities of logical residues. For coherent drift
the difference is stronger: a recovery map may average out the phase, turn
it into a logical rotation, or make it visible only conditionally on the
syndrome. We will call such a mismatch between the decoder and the actual
continuous noise decoder mismatch.

For a recovery map \(\cD\), denote by \(\cK(\cD)\) the space
\(\cK_{\rm code}\) built via \eqref{eq:k-code-definition}, and by
\(\widehat{\cK}(\cD)\) its observable factorization from Theorem
\ref{thm:krylov-visible-algorithm}. We say that \(\cD_2\) is a logical
post-processing of \(\cD_1\) if there exists a CPTP channel
\(\cM:\cB(\cH_L)\to\cB(\cH_L)\) such that
\[
    \cD_2=\cM\cD_1 .
\]

\begin{theorem}[Monotonicity of visible modes under post-processing]
\label{thm:recovery-comparison}
If \(\cD_2=\cM\cD_1\), then
\[
    \cK(\cD_2)\subseteq \cK(\cD_1).
\]
Consequently, the visible spectrum after logical post-processing is
contained in the visible reduction for \(\cD_1\). Post-processing
preserves, mixes, or hides the available modes. In particular, for any
finite family of admissible recovery maps there exists a recovery map
minimizing any quality functional \(J(\cD)\) with the property
\(J(\cM\cD)\le J(\cD)\) for a logical post-processing \(\cM\).
\end{theorem}

\begin{proof}
From \(\cD_2=\cM\cD_1\) it follows that
\[
    \cD_2^\dagger=\cD_1^\dagger\cM^\dagger .
\]
Since \(\cM^\dagger\) maps logical observables to logical observables,
each operator \(\cD_2^\dagger(F_j)\) lies in the linear span of the
operators \(\cD_1^\dagger(F_i)\). Applying all powers of \(\cL^\dagger\)
gives the inclusion \(\cK(\cD_2)\subseteq\cK(\cD_1)\). Passing to the
observable factorization can only further identify unobservable
directions. The last statement is finite-dimensional: on a finite family,
\(J\) attains its minimum.
\end{proof}

\begin{remark}
The partial order in the theorem is given by the logical post-processing
relation. For incomparable recovery maps, the visible spaces can also be
incomparable. Hence the choice of recovery map enters the physical problem
formulation, especially in the presence of coherent Hamiltonian noise
components.
\end{remark}

\subsection{Short time and the Knill--Laflamme conditions}

The short-time reduction must be consistent with the standard
correctability criterion. For \eqref{eq:lindblad}, the physical channel
has the expansion
\[
    e^{t\cL}(\rho)
    =
    K_0(t)\rho K_0(t)^\dagger
    +
    \sum_\alpha K_\alpha(t)\rho K_\alpha(t)^\dagger
    +
    O(t^2),
\]
where
\[
    K_0(t)
    =
    I-t\left(iH+\frac12\sum_\alpha L_\alpha^\dagger L_\alpha\right),
    \qquad
    K_\alpha(t)=\sqrt t\,L_\alpha .
\]
Consequently, the jump operators \(L_\alpha\) are first-order errors,
while \(K_0(t)\) contains the effective no-jump evolution.

\begin{proposition}[Consistency with first-order correctability]
\label{prop:first-order-correction}
Suppose that for the code projector \(P=VV^\dagger\) the conditions
\[
    P L_\alpha^\dagger L_\beta P=c_{\alpha\beta}P,
    \qquad
    P L_\alpha P=c_{\alpha 0}P
\]
hold for all \(\alpha,\beta\). Then there exists a recovery map \(\cD\)
such that the recovered logical channel has the form
\begin{equation}
\label{eq:first-order-logical}
    \Phi_t(\omega)
    =
    \omega
    -
    it[H_L,\omega]
    +
    O(t^2),
    \qquad
    H_L=V^\dagger P H P V .
\end{equation}
In particular, if \(H_L\) is proportional to \(I_L\) or is compensated by
a known logical frame, then
\[
    \Phi_t=\id_L+O(t^2).
\]
\end{proposition}

\begin{proof}
The conditions in the statement are the Knill--Laflamme conditions for the
set of errors \(\{I,L_\alpha\}\). Hence the jump branches
\(\sqrt t\,L_\alpha\) are correctable up to terms of order \(t^2\). The
anticommutator part \(K_0(t)\) contains
\(\sum_\alpha L_\alpha^\dagger L_\alpha\), which acts as a scalar on the
code by the same conditions and is logically neutral. The nonscalar first
order arises from the Hamiltonian \(PHP\). After identifying the code
subspace with \(\cH_L\), it gives the commutator in
\eqref{eq:first-order-logical}.
\end{proof}

This proposition is used below as a consistency check of the visible
reduction against the standard Knill--Laflamme criterion. The distinctive
role of the spectral language begins after the physical generator and
recovery map are chosen, when the same set of physical modes can produce
different logical channels.

\begin{remark}
This statement explains the familiar claim that a distance-\(d\) code
suppresses low-weight errors. In continuous time this means the
disappearance of the corresponding linear terms of logical decoherence.
Higher-order logical errors and oscillating modes are described by the
spectral reduction below.
\end{remark}

\section{Spectral signatures of coherent decoherence}
\label{sec:spectral-classification}

We now fix the code, the recovery map, and the physical generator. This
section uses the standard spectral analysis of the finite-dimensional
matrix
\[
    A=\cL^\dagger|_{\cK_{\rm code}}.
\]
The point of the classification is to apply this standard analysis only
to the observable realization of the logical channel. Under Hamiltonian
or coherently dissipative noise, complex visible eigenvalues are a
spectral signature of phase information that has reached the logical
level. The logical channel sees precisely those spectral modes that are
present in the decomposition of the vectors \(\cD^\dagger(F_j)\) and that
have nonzero pairing with the encoded states.

\begin{definition}
For an initial logical state \(\omega\), call an eigenvalue
\(\lambda\in\sigma(A)\) \emph{visible} if there exists a logical
observable \(F_j\) and a generalized eigencomponent \(X_{\lambda,m}\),
\(m=0,\ldots,\nu_\lambda-1\), where \(\nu_\lambda\) is the size of the
corresponding Jordan block, in the Jordan decomposition
\[
    \cD^\dagger(F_j)
    =
    \sum_{\lambda\in\sigma(A)}\sum_{m=0}^{\nu_\lambda-1}X_{\lambda,m}
\]
of \(\cD^\dagger(F_j)\) along generalized eigenspaces of \(A\), such that
for some such \(m\)
\[
    \operatorname{Tr}
    \left[
        X_{\lambda,m} V\omega V^\dagger
    \right]\neq0 .
\]
Denote the set of such eigenvalues by
\(\sigma_{\rm vis}(\omega)\).
\end{definition}

\begin{proposition}[Criterion for coherent logical oscillations]
\label{prop:general-code-oscillation}
Every matrix element of the logical channel \(\Phi_t\) is a finite linear
combination of functions of the form
\[
    t^m e^{\lambda t},
    \qquad
    \lambda\in\sigma(A),
\]
where \(m\) is smaller than the size of the corresponding Jordan block.
Sinusoidal components of the logical channel arise if and only if
\(\sigma_{\rm vis}(\omega)\) contains a complex eigenvalue \(a+ib\) with
\(b\neq0\). These components decay for \(a<0\) and keep a constant
amplitude for \(a=0\).
\end{proposition}

In Pauli models, after full phase averaging the logical channel is
usually described by real stochastic dynamics of probabilities. The
coherent part of the noise shows up differently: it can produce complex
visible modes, oscillations of logical fidelity, and revivals of memory
quality. In this sense the visible spectrum serves as a diagnostic of
which part of the continuous phase dynamics survived encoding, recovery,
and decoding.

\begin{proof}
By Proposition \ref{prop:visible-subspace}, all matrix elements of
\(\Phi_t\) are computed using \(e^{tA}\). The Jordan form of the
finite-dimensional matrix \(A\) gives an expansion of \(e^{tA}\) as a sum
of terms \(t^m e^{\lambda t}\). For non-real \(\lambda=a\pm ib\), real
linear combinations give factors \(e^{at}\cos bt\) and \(e^{at}\sin bt\).
If all visible \(\lambda\) are real, no such factors appear.
\end{proof}

\begin{theorem}[{Classification of visible spectral contributions}]
\label{thm:strict-visible-classification}
For a fixed initial logical state \(\omega\), the code-visible spectrum
\(\sigma_{\rm vis}(\omega)\) decomposes into three elementary types of
visible spectral contributions, which may occur simultaneously: an
\textbf{aperiodic contribution}, from a real
\(\lambda\in\sigma_{\rm vis}(\omega)\); a \textbf{damped oscillatory
contribution}, from a pair \(\lambda=a\pm ib\) with \(a<0\), \(b\neq0\);
and an \textbf{undamped oscillatory contribution}, from a pair
\(\lambda=\pm ib\) with \(b\neq0\). In general \(\sigma_{\rm
vis}(\omega)\) may contain several such pairs at once, so the logical
channel can display aperiodic relaxation together with damped and
undamped oscillations; the classification refers to which of the three
types are present, not to a single exclusive regime. Boundedness of the
CPTP semigroup excludes growing logical oscillations: if all nonzero
visible eigenvalues have negative real part, then any oscillatory
contribution decays, while visible eigenvalues exactly on the imaginary
axis produce an undamped contribution that persists asymptotically.
\end{theorem}

\begin{proof}
Since \(e^{t\cL}\) is a CPTP semigroup, the adjoint semigroup
\(e^{t\cL^\dagger}\) preserves the identity operator and positivity.
Hence it is bounded in operator norm. The restriction to the invariant
subspace \(\cK_{\rm code}\) is also bounded. In finite dimension this
implies
\[
    \Re\lambda\le0,
    \qquad
    \lambda\in\sigma(A),
\]
and the Jordan blocks on the imaginary axis are semisimple. This rules out
exponentially growing components. The classification follows from
Proposition \ref{prop:general-code-oscillation}.
\end{proof}

\subsection{Operational significance of the visible modes}

Let \(d_L=\dim\cH_L\). Then the average fidelity of the logical channel is
expressed through the entanglement fidelity:
\[
    F_{\rm av}(t)=\frac{d_L F_e(t)+1}{d_L+1}.
\]
Since \(F_e(t)\) is a linear function of the Choi matrix of \(\Phi_t\), it
depends only on the same code-visible modes. For a logical Pauli channel
\[
    \Phi_t(\omega)
    =
    \sum_{\ell} p_\ell(t)\,
    \overline P_\ell \omega \overline P_\ell^\dagger
\]
the entanglement fidelity equals the probability of a trivial logical
residue:
\[
    F_e(t)=p_I(t).
\]

Hence logical oscillations appear as fluctuations of the probability of
successful storage, of the average fidelity, and of the anisotropy of
transmission along different logical directions.

\section{Stabilizer syndrome reduction}
\label{sec:stabilizer-reduction}

An important class of models arises when the physical Lindblad-type noise
is Pauli noise:

\begin{equation}
\label{eq:pauli-lindblad}
    \cL_P(\rho)
    =
    \sum_{a\in A}
    \gamma_a
    \left(
        P_a\rho P_a-\rho
    \right),
    \qquad
    \gamma_a\ge0,
\end{equation}
where \(P_a\in\cP_n\). A Pauli channel is a quantum channel that is a
stochastic mixture of Pauli operators. For a single qubit it has the form
\[
    \cE(\rho)
    =
      p_I\rho+p_X\,X\rho X+p_Y\,Y\rho Y+p_Z\,Z\rho Z,
    \qquad
    p_I+p_X+p_Y+p_Z=1,
\]
where \(p_I,p_X,p_Y,p_Z\ge0\). For \(n\) qubits, analogously,
\[
    \cE(\rho)
    =
    \sum_{P\in\cP_n}p_P\,P\rho P^\dagger,
    \qquad
    \sum_{P\in\cP_n}p_P=1.
\]
Equivalently, with probability \(p_P\) a particular Pauli string \(P\) is
applied to the state. Such a generator preserves the class of Pauli
channels. Preservation of the class of Pauli channels means that the
channel \(e^{t\cL_P}\), generated by this generator, has for every \(t\)
the form of a stochastic mixture of Pauli errors. Under such dynamics,
only the probabilities \(p_P(t)\) of individual Pauli labels change; a
commutator term of the form \(-i[H,\rho]\), corresponding to a coherent
rotation, does not arise. If at the initial moment the Pauli label of the
error equals the identity operator \(I\), i.e., there is no physical Pauli
error before the noise acts, then at time \(t\) the physical noise can be
represented as a random Pauli operator \(G\) with distribution \(q_G(t)\).

The distribution \(q(t)\) satisfies a finite Markov chain on the group of
Pauli labels without phases:
\begin{equation}
\label{eq:pauli-random-walk}
    \dot q_G(t)
    =
    \sum_{a\in A}
    \gamma_a
    \bigl(
        q_{P_aG}(t)-q_G(t)
    \bigr).
\end{equation}
After stabilizer recovery, the physical label \(G\) is replaced by the
logical residue \(\ell(G)\) defined in
\eqref{eq:logical-residue}. Hence the logical probabilities are
\begin{equation}
\label{eq:logical-probabilities-from-pauli}
    p_\ell(t)
    =
    \sum_{G:\ell(G)=\ell} q_G(t).
\end{equation}

\begin{theorem}[Reduction of Pauli-Lindblad noise to a logical Pauli channel]
\label{thm:pauli-chain-reduction}
Let a stabilizer code, a fixed Pauli recovery map \(R_s\), and a
Pauli-Lindblad generator \eqref{eq:pauli-lindblad} be given. Then the
effective logical channel has the form
\begin{equation}
\label{eq:logical-pauli-channel}
    \Phi_t(\omega)
    =
    \sum_{\ell\in\cN(\cS)/\cS}
    p_\ell(t)\,
    \overline P_\ell \omega \overline P_\ell^\dagger ,
\end{equation}
where \(p_\ell(t)\) is given by \eqref{eq:logical-probabilities-from-pauli}.
If the partition of Pauli labels into classes \(\ell(G)\) is strongly
lumpable for the chain \eqref{eq:pauli-random-walk}, then the vector
\(p(t)=(p_\ell(t))\) satisfies the autonomous equation
\[
    \dot p(t)=Q_{\rm log}p(t)
\]
with the generator of a finite Markov chain on the logical Pauli classes.
\end{theorem}

\begin{proof}
The Pauli-Lindblad generator \eqref{eq:pauli-lindblad} produces a random
multiplication of the current Pauli error by one of the operators
\(P_a\), giving \eqref{eq:pauli-random-walk}. After syndrome measurement
and correction \(R_s\), the residual operator lies in the normalizer of
the stabilizer and acts on the code as a logical Pauli operator
\(\overline P_\ell\). Summing all physical labels with the same residue
gives \eqref{eq:logical-pauli-channel}. If the partition is lumpable, then
standard finite Markov chain theory gives a closed equation for
\(p_\ell(t)\).
\end{proof}

\begin{remark}
In general, the logical channel is given by
\eqref{eq:logical-probabilities-from-pauli}, and computing \(p(t)\)
requires a finer distribution over physical Pauli labels or over
syndrome-logical classes. This is a direct analogue of the appearance of
hidden coherences in reduced projections.
\end{remark}

\begin{corollary}[Detailed balance excludes classical oscillations]
\label{cor:detailed-balance-no-oscillations}
If the closed generator \(Q_{\rm log}\) satisfies detailed balance with
respect to a positive stationary distribution \(\pi\),
\[
    \pi_\ell (Q_{\rm log})_{m\ell}
    =
    \pi_m (Q_{\rm log})_{\ell m},
\]
then the spectrum of \(Q_{\rm log}\) is real. Consequently, the logical
probabilities \(p_\ell(t)\) contain no damped sinusoidal components.
\end{corollary}

\begin{proof}
Under detailed balance, the matrix \(Q_{\rm log}\) is self-adjoint with
respect to the inner product
\[
    \langle x,y\rangle_\pi=\sum_\ell \pi_\ell^{-1}\overline{x_\ell}y_\ell .
\]
Hence it is diagonalizable with real spectrum. Then only real exponentials
appear in the solutions of \(\dot p=Q_{\rm log}p\).
\end{proof}

\begin{remark}[Toward topological codes]
The same procedure applies to the surface code and other stabilizer
families: one needs to specify \(\cL^\dagger\), the set of decodable
logical observables, and the specific decoder/recovery map. The exact
dimension of \(\cK_{\rm code}\) depends on the noise and on the chosen
decoder; for local Pauli-Lindblad models it is often reduced by the
syndrome structure and locality, whereas coherent non-Pauli terms can
enlarge the Krylov closure via syndrome coherences. The method thus does
not give a ready-made threshold formula, but rather an algorithmic way to
determine which continuous modes of the physical noise are visible in the
logical channel of a specific decoder.

Concretely, for the surface code the required input data are: a local
Pauli-Lindblad or Hamiltonian generator \(\cL\) supported on the qubits
of the lattice; a syndrome decoder such as minimum-weight perfect
matching or union-find, which fixes the recovery map \(\cD\); and the
pair of logical observables \(\overline X,\overline Z\) given by the
standard homological string representatives, which fixes the initial
space \(W_0\) in \eqref{eq:k-code-definition}. Building the Krylov
closure \(\cK_{\rm code}\) of \(\{\cD^\dagger(\overline
X),\cD^\dagger(\overline Z)\}\) under \(\cL^\dagger\) would then single
out exactly the decoder-visible syndrome coherences for that particular
decoder, without by itself producing a threshold estimate.
\end{remark}

\section{Effect of Pauli noise and coherent drift}
\label{sec:examples}

\subsection{Three-qubit repetition code against bit-flip noise}

Consider the code
\[
    \ket{0_L}=\ket{000},
    \qquad
    \ket{1_L}=\ket{111},
\]
correcting a single error \(X_i\). Suppose each physical qubit is subject
to independent bit-flip Lindblad noise
\[
    \cL_X(\rho)
    =
    \gamma\sum_{i=1}^{3}
    (X_i\rho X_i-\rho).
\]
For a single qubit, the probability of a flip error over time \(t\) is
\[
    p(t)=\frac{1-e^{-2\gamma t}}2.
\]
After a single majority-vote recovery, a logical error occurs if and only
if two or three physical qubits have flipped:
\begin{equation}
\label{eq:repetition-logical-error}
    p_L(t)
    =
    3p(t)^2(1-p(t))+p(t)^3
    =
    \frac12-\frac34e^{-2\gamma t}
    +\frac14e^{-6\gamma t}.
\end{equation}
Consequently,
\[
    p_L(t)=3\gamma^2t^2+O(t^3).
\]
The linear term vanishes: single errors are correctable. The effective
logical channel is
\[
    \Phi_t(\omega)
    =
    (1-p_L(t))\omega
    +
    p_L(t)\,\overline X\omega \overline X .
\]
In this example the spectrum of the code-visible reduction is real, so
the logical channel relaxes aperiodically.

\begin{remark}
If recovery is repeated at intervals \(\Delta t\), the probability of a
logical error per cycle is
\[
    p_L(\Delta t)=3\gamma^2\Delta t^2+O(\Delta t^3).
\]
Over a fixed time \(T\), after \(T/\Delta t\) cycles the leading
contribution scales as \(3\gamma^2T\Delta t\). This is a simple dynamical
form of first-order suppression.
\end{remark}

\subsection{Steane code: coherent drift on seven qubits}

Consider the \([[7,1,3]]\) Steane code \cite{Steane1996} and only the
\(X\)-part of its recovery. The syndrome of a single \(X_j\) error is
given by the nonzero column \(h_j\in\mathbb F_2^3\) of the Hamming
matrix; the seven columns \(h_j\) run through all nonzero elements of
\(\mathbb F_2^3\). For a subset \(E\subset\{1,\ldots,7\}\), denote
\[
    \sigma(E)=\bigoplus_{j\in E} h_j .
\]
The standard minimum-weight decoder corrects a syndrome \(\sigma\ne0\)
with the operator \(X_\sigma\), i.e., \(X_j\) with \(h_j=\sigma\), and
leaves a zero syndrome uncorrected.

Suppose that over one cycle the physical noise is a coherent \(X\)-drift
\[
    U_\theta
    =
    \exp\left[
        -\frac{i\theta}{2}\sum_{j=1}^7 X_j
    \right]
    =
    \prod_{j=1}^7(cI-isX_j),
    \qquad
    c=\cos\frac{\theta}{2},\quad
    s=\sin\frac{\theta}{2}.
\]
In the language of Section \ref{sec:lindblad-code}, this one-cycle unitary
is generated by the Hamiltonian drift
\(H_\theta=\frac{\theta}{2\tau}\sum_{j=1}^7 X_j\) over a cycle of duration
\(\tau\), so that \(U_\theta=e^{\tau\cL_\theta}\) with
\(\cL_\theta(\rho)=-i[H_\theta,\rho]\), and
\(\Phi_{\theta,{\rm St}}^{\rm coh}=\cD e^{\tau\cL_\theta}\cE\) is again an
instance of the effective logical channel \(\Phi_t\). The Kraus operators
computed below are the syndrome components of the pairing
\(\cD^\dagger(F_j)\) from \eqref{eq:heisenberg-code-pairing}, expanded
order by order in \(\theta\) rather than in continuous time; each power of
\(\theta\) corresponds to one step of the Krylov closure
\eqref{eq:k-code-definition} building \(\cK_{\rm code}\).

After syndrome measurement and correction, the Kraus operator of the
logical channel for syndrome \(\sigma\) is obtained by summing over all
physical error subsets with the given syndrome. The leading terms are
\[
    K_0
    =
    c^7 I_L+7is^3c^4\,\overline X+O(s^4),
\]
since the zero syndrome is shared by the seven weight-three words of the
classical Hamming code, giving the logical \(\overline X\). For each
nonzero syndrome \(\sigma\), minimum-weight correction gives
\[
    K_\sigma
    =
    -isc^6 I_L-3s^2c^5\,\overline X+O(s^3).
\]
The three terms in the coefficient of \(\overline X\) correspond to the
three pairs of qubits \(i,j\) for which \(h_i\oplus h_j=\sigma\); after
correction, \(X_\sigma X_iX_j\) is a weight-three representative of the
logical \(\overline X\).

Summing the contributions of all syndromes, we obtain for the coherent
channel
\begin{equation}
\label{eq:steane-coherent-leading}
    \Phi_{\theta,{\rm St}}^{\rm coh}(\omega)
    =
    \Phi_{\theta,{\rm St}}^{\rm twirl}(\omega)
    -14i\,s^3c^{11}[\overline X,\omega]
    +O(s^4).
\end{equation}
Hence, for small angle,
\[
    \Phi_{\theta,{\rm St}}^{\rm coh}(\omega)
    =
    \Phi_{\theta,{\rm St}}^{\rm twirl}(\omega)
    -i\frac{7}{4}\theta^3[\overline X,\omega]
    +O(\theta^4).
\]
If the same physical rotation is instead replaced by a Pauli-twirled
channel with probability \(p=s^2\) on each qubit, the minimum-weight
decoder gives a leading stochastic logical error rate
\[
    p_{L,{\rm St}}^{\rm twirl}
    =
    21s^4+O(s^6)
    =
    \frac{21}{16}\theta^4+O(\theta^6),
\]
since each of the \(21\) pairs of physical \(X\)-errors is corrected into
a weight-three logical \(\overline X\). Thus, already for the seven-qubit
stabilizer code, coherent drift produces a logical Hamiltonian rotation of
lower order than the stochastic logical-error probability in the
Pauli-twirled model. The visible spectral reduction captures precisely
this commutator contribution, whereas the reduction to Pauli-error
probabilities discards it.

The commutator term \(-14is^3c^{11}[\overline X,\omega]\) in
\eqref{eq:steane-coherent-leading} is exactly a code-visible
Hamiltonian-type contribution to \(A=\cL_\theta^\dagger|_{\cK_{\rm code}}\)
in the sense of Section \ref{sec:spectral-classification}: it is the
lowest Krylov order at which \(\cD^\dagger(\overline X)\) acquires a
nonzero antisymmetric (purely imaginary) component, i.e., a visible
complex-conjugate pair of eigenvalues of \(A\) rather than a real,
aperiodic one. This is the same spectral signature classified in
Proposition \ref{prop:general-code-oscillation} and Theorem
\ref{thm:strict-visible-classification}, here evaluated order by order in
\(\theta\); it is absent from the Pauli-twirled reduction of Section
\ref{sec:stabilizer-reduction}, whose generator \(Q_{\rm log}\) only
produces real relaxation rates.

\subsection{Two-qubit decoherence-free subspace}

Consider the code
\[
    \ket{0_L}=\ket{01},
    \qquad
    \ket{1_L}=\ket{10}.
\]
Under collective dephasing with operator
\[
    J_z=Z_1+Z_2
\]
the generator has the form
\[
    \cL_{\rm col}(\rho)
    =
    \gamma
    \left(
        J_z\rho J_z
        -
        \frac12\{J_z^2,\rho\}
    \right).
\]
On both code vectors, \(J_z\) has the same eigenvalue zero. Hence
\[
    \cL_{\rm col}(V\omega V^\dagger)=0
\]
for all \(\omega\). The code-visible space contains only the zero
generator, and the logical channel is the identity:
\[
    \Phi_t=\id_L .
\]

For independent dephasing
\[
    \cL_{\rm ind}(\rho)
    =
    \gamma\sum_{i=1}^{2}
    (Z_i\rho Z_i-\rho),
\]
on the code subspace both operators \(Z_i\) act as the logical
\(\overline Z\) up to sign. Hence
\[
    \Phi_t(\omega)
    =
    \frac{1+e^{-4\gamma t}}2\,\omega
    +
    \frac{1-e^{-4\gamma t}}2\,
    \overline Z\omega\overline Z .
\]
Collective symmetry thus singles out the DFS regime. Under independent
dephasing, an ordinary logical dephasing channel arises.

\subsection{A single physical drift and decoder mismatch}

The next example shows the dependence of the visible modes on the
recovery map. One and the same physical Liouvillian with undamped
oscillating modes gives an identity logical channel for one recovery map
and an oscillating logical channel for another. In decoder language this
is a minimal model of decoder mismatch: the physical drift occurs in a
syndrome degree of freedom, and the chosen logical frame decides whether
this drift turns into a logical error.

\begin{proposition}[One drift, two visible pictures]
\label{prop:same-drift-two-recoveries}
Let the physical space decompose as
\[
    \cH_P=\mathbb C^2_{\rm syn}\otimes\cH_L,
\]
with the encoding given by embedding into the zero syndrome:
\[
    V\ket{\psi}=\ket0_{\rm syn}\otimes\ket{\psi}.
\]
Consider the Hamiltonian physical dynamics
\[
    \cL_{\rm leak}(\rho)=-i[H_{\rm leak},\rho],
    \qquad
    H_{\rm leak}
    =
    \Omega
    (\,\ket0\bra1+\ket1\bra0\,)_{\rm syn}\otimes I_L .
\]
Then the full physical Liouvillian has purely imaginary eigenvalues. For
the recovery map
\[
    \cD_+(\rho)
    =
    \sum_{s=0}^{1}
    (\bra s\otimes I_L)\rho(\ket s\otimes I_L)
\]
the effective logical channel is the identity:
\[
    \Phi_t^{(+)}=\id_L .
\]
For another recovery map
\[
    \cD_Z(\rho)
    =
    (\bra0\otimes I_L)\rho(\ket0\otimes I_L)
    +
    Z_L(\bra1\otimes I_L)\rho(\ket1\otimes I_L)Z_L
\]
one obtains the logical Pauli channel
\[
    \Phi_t^{(Z)}(\omega)
    =
    \cos^2(\Omega t)\,\omega
    +
    \sin^2(\Omega t)\,Z_L\omega Z_L .
\]
Consequently, the same physical mode has a zero logical image for
\(\cD_+\) and a nonzero oscillating image for \(\cD_Z\), with frequency
\(2\Omega\).
\end{proposition}

\begin{proof}
Unitary evolution maps the encoded vector to
\[
    e^{-itH_{\rm leak}}
    (\ket0_{\rm syn}\otimes\ket{\psi})
    =
    \bigl(\cos(\Omega t)\ket0_{\rm syn}
    -i\sin(\Omega t)\ket1_{\rm syn}\bigr)
    \otimes\ket{\psi}.
\]
For an arbitrary logical density matrix \(\omega\), after evolution the
diagonal syndrome blocks equal
\[
    \cos^2(\Omega t)\,\omega,
    \qquad
    \sin^2(\Omega t)\,\omega,
\]
while the off-diagonal blocks are proportional to
\(\cos(\Omega t)\sin(\Omega t)\omega\). Both recovery maps given in the
statement measure the syndrome and therefore discard the off-diagonal
syndrome blocks. Recovery \(\cD_+\) adds the two diagonal branches with the
same logical frame, giving
\[
    (\cos^2\Omega t+\sin^2\Omega t)\omega=\omega.
\]
Recovery \(\cD_Z\) applies an additional logical \(Z_L\) to the branch
\(s=1\), from which the stated formula for \(\Phi_t^{(Z)}\) follows.
Finally, \(H_{\rm leak}\) has levels \(\pm\Omega\), so the commutator
Liouvillian contains frequencies \(\pm2\Omega\).
\end{proof}

\begin{remark}
The recovery maps \(\cD_+\) and \(\cD_Z\) are incomparable under the
logical post-processing relation from Theorem \ref{thm:recovery-comparison}.
After \(\cD_+\), information about the syndrome branch is erased, whereas
\(\cD_Z\) uses this branch to select the logical Pauli frame. Hence a
Pauli description of the syndrome probability alone is not enough: what
matters is how the decoder maps the syndrome branch onto the logical
frame.
\end{remark}

\subsection{Coherently dissipative logical oscillation}

Suppose that after code reduction and recovery, the effective logical
generator on a single qubit has the form
\begin{equation}
\label{eq:logical-driven-dephasing}
    \cL_L(\omega)
    =
    -i\frac{\Omega}{2}[\overline X,\omega]
    +
    \Gamma(\overline Z\omega\overline Z-\omega).
\end{equation}
In Bloch-vector coordinates \(\omega=(I+r_x\overline X+r_y\overline
Y+r_z\overline Z)/2\) we obtain
\[
    \dot r_x=-2\Gamma r_x,
    \qquad
    \frac{d}{dt}
    \begin{pmatrix}
        r_y\\ r_z
    \end{pmatrix}
    =
    \begin{pmatrix}
        -2\Gamma & -\Omega\\
        \Omega & 0
    \end{pmatrix}
    \begin{pmatrix}
        r_y\\ r_z
    \end{pmatrix}.
\]
The eigenvalues of the second block are
\[
    -\Gamma\pm\sqrt{\Gamma^2-\Omega^2}.
\]
For \(\Omega>\Gamma\) they form a complex conjugate pair
\[
    -\Gamma\pm i\sqrt{\Omega^2-\Gamma^2}.
\]
Consequently, the logical channel contains damped oscillations. The
mechanism is the noncommutativity of the logical Hamiltonian and the
logical dephasing.

\subsection{Non-Pauli noise and full visible reduction}

Amplitude damping with jumps
\[
    \sigma_-^{(i)}=\ket0\bra1_i
\]
gives non-Pauli noise. Even for a stabilizer code, such jumps create
superpositions of Pauli components and can mix syndrome coherences with
logical observables. In this situation the exact reduction is given by
the full space of decodable observables.

The space \(\cK_{\rm code}\) from \eqref{eq:k-code-definition} shows which
non-Pauli coherences become logically visible after the chosen recovery
map, and which belong to the kernel of the logical projection.

\section{Quantum memory and visible storage modes}
\label{sec:quantum-memory}

Quantum memory is the central operational regime of a quantum
error-correcting code. Its task is to preserve an unknown logical state
for a time \(T\). The natural object of analysis is the storage channel
that connects the initial and final logical states.

Without intermediate correction, the memory channel over time \(T\)
coincides with the code projection already used:
\[
    \Phi_T^{\rm pass}
    =
    \cD e^{T\cL}\cE .
\]
This is passive memory. Active quantum memory differs in that recovery is
applied periodically. Let \(\tau\) be the duration of one cycle, and let
\(\cR:\cB(\cH_P)\to\cB(\cH_P)\) be the physical recovery map that measures
the syndrome, applies a correction, and returns the state to the code
region. After \(m\) cycles, the physical state is
\[
    \rho_m
    =
    (\cR e^{\tau\cL})^m\cE(\omega),
    \qquad
    T=m\tau .
\]
If \(\cD_0\) denotes the final decoding from the code space, the memory
logical channel has the form
\[
    \Phi_{m,\tau}^{\rm mem}
    =
    \cD_0(\cR e^{\tau\cL})^m\cE .
\]

In the idealized model of instantaneous recovery, a stronger
code-returning condition often holds:
\begin{equation}
\label{eq:code-returning-recovery}
    \cR e^{\tau\cL}\cE
    =
    \cE\Psi_\tau ,
    \qquad
    \Psi_\tau:=\cD_0\cR e^{\tau\cL}\cE .
\end{equation}
Then the entire memory is described by powers of a single-cycle logical
channel:
\[
    \Phi_{m,\tau}^{\rm mem}
    =
    \Psi_\tau^m .
\]
It is precisely here that quantum error correction connects to visible
Liouvillian modes: the code and the recovery map replace the physical
generator \(\cL\) by a discrete logical channel \(\Psi_\tau\), and the
long-time quality of the memory is determined by the visible spectrum of
\(\Psi_\tau\).

\begin{definition}
An eigenvalue \(\zeta\) of the channel \(\Psi_\tau^\dagger\) is called a
\emph{memory-visible mode} if the corresponding generalized eigencomponent
gives a nonzero contribution to some matrix element
\[
    \operatorname{Tr}[F\Psi_\tau^m(\omega)]
\]
for some logical observable \(F\) and some logical state \(\omega\).
\end{definition}

\begin{theorem}[Visible modes of quantum memory]
\label{thm:memory-visible-modes}
Suppose the code-returning condition \eqref{eq:code-returning-recovery}
holds. Then for any logical observable \(F\) and initial state \(\omega\),
the quantity
\[
    \operatorname{Tr}[F\Phi_{m,\tau}^{\rm mem}(\omega)]
\]
is a finite linear combination of terms
\[
    m^q \zeta^m,
    \qquad
    \zeta\in\sigma(\Psi_\tau^\dagger).
\]
Complex memory-visible \(\zeta=|\zeta|e^{i\theta}\) give discrete
oscillations and possible fidelity revivals. If all nontrivial visible
eigenvalues satisfy \(|\zeta|<1\), the memory loses the corresponding
logical information exponentially in the number of cycles. The dominant
visible modulus
\[
    \eta_\tau
    =
    \max\{|\zeta|\colon \zeta\neq1,\ \zeta
    \text{ is memory-visible}\}
\]
sets the leading storage-time scale
\begin{equation}
\label{eq:memory-time-visible-mode}
    T_{\rm mem}(\varepsilon)
    \simeq
    \frac{\tau\log(C/\varepsilon)}{-\log\eta_\tau},
\end{equation}
if the dominant visible part is diagonalizable and one wants to suppress
its contribution to a level \(\varepsilon\).
\end{theorem}

\begin{proof}
By induction, \eqref{eq:code-returning-recovery} gives
\[
    (\cR e^{\tau\cL})^m\cE=\cE\Psi_\tau^m,
\]
so the final logical channel equals \(\Psi_\tau^m\). The remaining
statement is a Jordan decomposition of the finite-dimensional operator
\(\Psi_\tau^\dagger\). Since \(\Psi_\tau\) is a CPTP channel,
\(\Psi_\tau^\dagger\) is unital and positive, so its spectrum lies in the
unit disk, and the Jordan blocks on the circle \(|\zeta|=1\) are
semisimple. If the dominant visible contribution has modulus
\(\eta_\tau<1\) and is diagonalizable, its magnitude is bounded as
\(C\eta_\tau^m\). Substituting \(m=T/\tau\) gives
\eqref{eq:memory-time-visible-mode}.
\end{proof}

\begin{proposition}[How error correction increases memory time]
\label{prop:error-correction-memory-scaling}
Let the single-cycle logical channel have the small-step expansion
\begin{equation}
\label{eq:cycle-small-step}
    \Psi_\tau
    =
    \id_L+\tau^{r+1}\mathcal G+O(\tau^{r+2}),
    \qquad
    r\ge0,
\end{equation}
and let \(\mu_j\) be the visible eigenvalues of \(\mathcal G^\dagger\).
Then the corresponding eigenvalues of \(\Psi_\tau^\dagger\) have the form
\[
    \zeta_j(\tau)
    =
    1+\tau^{r+1}\mu_j+O(\tau^{r+2}),
\]
and the effective logical decay rates per physical time equal
\[
    -\frac1{\tau}\log|\zeta_j(\tau)|
    =
    -\tau^r\Re\mu_j+O(\tau^{r+1}).
\]
Consequently, if the recovery map suppresses the single-cycle logical
error to order \(O(\tau^{r+1})\), the leading rate of memory loss decreases
as \(O(\tau^r)\).
\end{proposition}

\begin{proof}
The formula for \(\zeta_j(\tau)\) follows from ordinary finite-dimensional
eigenvalue perturbation theory applied to
\eqref{eq:cycle-small-step}. Next,
\[
    \log|1+\tau^{r+1}\mu_j+O(\tau^{r+2})|
    =
    \tau^{r+1}\Re\mu_j+O(\tau^{r+2}),
\]
and dividing by \(\tau\) gives the stated rate. If, without correction,
the logical channel has a linear error per cycle, then \(r=0\). If the
code and recovery map remove the entire first order and leave a logical
error probability of \(O(\tau^2)\), as in the three-qubit repetition code
of \eqref{eq:repetition-logical-error}, then \(r=1\), and the effective
memory decay rate is proportional to \(\tau\).
\end{proof}

\section{Asymptotic logical channels}
\label{sec:asymptotics}

Suppose the physical semigroup has a limiting projection onto a set of
asymptotic states, or, more generally, suppose that for a given code the
limit
\[
    \Phi_\infty
    =
    \lim_{t\to\infty}\cD e^{t\cL}\cE
\]
exists in channel norm. Then the asymptotics distinguishable by the code
is described by the image of the stationary physical asymptotics under
\(\cD\), restricted to encoded inputs.

\begin{theorem}[Classification of limiting logical channels]
\label{thm:asymptotic-code-channels}
Suppose the limit \(\Phi_\infty=\lim_{t\to\infty}\Phi_t\) exists. Then
\(\Phi_\infty\) is a CPTP channel on the logical system and is completely
determined by the zero visible modes of \(A=\cL^\dagger|_{\cK_{\rm code}}\);
if \(A\) has no nonzero purely imaginary visible eigenvalues, then all
decaying modes vanish and the limiting channel is a projection onto the
visible stationary space; moreover, the set of possible limiting logical
channels has dimension bounded above by \(d_L^4-d_L^2\), i.e., by the
dimension of the affine space of CPTP channels.
\end{theorem}

\begin{proof}
Each \(\Phi_t\) is a composition of CPTP channels. Hence the limit in
norm is also a CPTP channel. By Proposition \ref{prop:visible-subspace},
all matrix elements of \(\Phi_t\) are given by \(e^{tA}\). If the limit
exists, then in the Jordan decomposition only components with
\(\lambda=0\) survive; decaying modes vanish, and nonzero purely
imaginary visible modes are incompatible with the existence of the limit.
The last dimension bound is the standard estimate for the space of linear
maps \(\cB(\cH_L)\to\cB(\cH_L)\) subject to trace-preservation
conditions.
\end{proof}

Typical limiting regimes are as follows: \textbf{ideal memory}, for which
\(\Phi_\infty=\id_L\); \textbf{logical Pauli noise},
\[
    \Phi_\infty(\omega)
    =
    \sum_\ell p_\ell^\infty
    \overline P_\ell\omega\overline P_\ell^\dagger;
\]
\textbf{partial loss of logical information}, e.g., a dephasing
projection preserving only the diagonal in some logical basis; and
\textbf{peripheral dynamics}, in which the logical channel periodically
changes the logical frame.

\section{Conclusion}

Coherent noise creates a gap between continuous physical dynamics and the
discrete Pauli models used in fast estimates of logical errors. A quantum
code together with a chosen recovery map defines an output map for the
physical Lindbladian dynamics. This map takes the full semigroup
\[
    e^{t\cL}
    \mapsto
    \Phi_t=\cD e^{t\cL}\cE ,
\]
that is, it defines a linear system with the logical channel playing the
role of the observed output. Hence the code-visible modes are the QEC
version of the observable modes of a finite-dimensional system.

The technical reduction consists in constructing the Krylov closure of
the decodable logical observables, factoring out the unobservable
subspace, and analyzing the induced operator
\[
    \cL^\dagger|_{\cK_{\rm code}}.
\]
The computational point of the reduction is to replace the full physical
evolution with a matrix on \(\widehat{\cK}_{\rm code}\) that contains
exactly the modes affecting the logical channel. The short-time expansion
is consistent with the Knill--Laflamme conditions: correctable Lindblad
errors give zero dissipative first order after recovery. The further
spectral classification is an application of the ordinary Jordan
decomposition to the minimal observable realization of the logical
channel; complex visible eigenvalues are a signature of phase dynamics
that survived encoding and decoding.

The operational contribution of this formulation is that the recovery map
enters the problem as a spectral filter. Logical post-processing can hide
part of the visible modes, and incomparable recovery maps can select
different frequencies of the same physical generator. This gives a
linear-algebraic model of decoder mismatch for coherent noise. For
stabilizer codes under Pauli noise, the reduction takes the form of a
finite Markov problem on syndrome-logical classes; for non-Pauli and
coherent noise, the exact reduction goes through the full space of
decodable observables.

For quantum memory, the same language gives a discrete observable
realization: active memory with correction cycles is described by powers
of the logical channel \(\Psi_\tau\). The storage time is determined by
the dominant visible eigenvalue modulus of this channel, and suppressing
the single-cycle logical error to higher order in \(\tau\) reduces the
visible decay rates. In this form, the method serves as a
linear-algebraic tool for comparing codes, recovery maps, and noise
models at the level of the logical channel.

\bibliography{decoherence_refs}

\end{document}